\documentclass[11pt]{article}%
\usepackage{amsmath}
\usepackage{amsfonts}
\usepackage{amssymb}
\usepackage{graphicx}
\usepackage{fullpage}%
\providecommand{\U}[1]{\protect\rule{.1in}{.1in}}
\newtheorem{theorem}{Theorem}

\newtheorem{definition}[theorem]{Definition}

\newtheorem{lemma}[theorem]{Lemma}

\newtheorem{problem}[theorem]{Problem}
\newtheorem{proposition}[theorem]{Proposition}

\newenvironment{proof}[1][Proof]{\noindent\textbf{#1.} }{\ \rule{0.5em}{0.5em}}
\begin{document}

\title{Impossibility of Succinct Quantum Proofs for Collision-Freeness}
\author{Scott Aaronson\thanks{MIT. \ Email: aaronson@csail.mit.edu. \ \ This material
is based upon work supported by the National Science Foundation under Grant
No. 0844626. \ Also supported by a DARPA YFA grant and a Sloan Fellowship.}}
\date{}
\maketitle

\begin{abstract}
We show that any quantum algorithm to decide whether a function $f:\left[
n\right]  \rightarrow\left[  n\right]  $ is a permutation or far from a
permutation\ must make $\Omega\left(  n^{1/3}/w\right)  $ queries to $f$, even
if the algorithm is given a $w$-qubit quantum witness in support of $f$ being
a permutation. \ This implies that there exists an oracle $A$ such that
$\mathsf{SZK}^{A}\not \subset \mathsf{QMA}^{A}$, answering an eight-year-old
open question of the author. \ Indeed, we show that relative to some oracle,
$\mathsf{SZK}$ is not in the counting class $\mathsf{A}_{\mathsf{0}%
}\mathsf{PP}$ defined by Vyalyi. \ The proof is a fairly simple extension of
the quantum lower bound for the collision problem.

\end{abstract}

\section{Introduction\label{INTRO}}

The \textit{collision problem} is to decide whether a black-box function
$f:\left[  n\right]  \rightarrow\left[  n\right]  $\ is one-to-one (i.e., a
permutation) or two-to-one function, promised that one of these is the case.
\ Together with its close variants, the collision problem is one of the
central problems studied in quantum computing theory; it abstractly models
numerous other problems such as graph isomorphism and the breaking of
cryptographic hash functions.

In this paper, we will mostly deal with a slight generalization of the
collision problem that we call the \textit{Permutation Testing Problem}, or
PTP. \ This is a property testing problem, in which we are promised that
$f:\left[  n\right]  \rightarrow\left[  n\right]  $ is either a permutation or
far from any permutation, and are asked to decide which is the case.

In 1997, Brassard, H\o yer, and Tapp \cite{bht} gave a quantum algorithm for
the collision problem that makes $O\left(  n^{1/3}\right)  $\ queries to $f$,
an improvement over the $\Theta\left(  \sqrt{n}\right)  $\ randomized query
complexity that follows from the birthday paradox. \ Brassard et al.'s
algorithm is easily seen to work for the PTP as well.

Five years later, Aaronson \cite{aar:col} proved the first non-constant lower
bound for these problems: namely, any bounded-error quantum algorithm to solve
them needs $\Omega\left(  n^{1/5}\right)  $\ queries to $f$. \ Aaronson and
Shi \cite{as} subsequently improved the lower bound to $\Omega\left(
n^{1/3}\right)  $, for functions $f:\left[  n\right]  \rightarrow\left[
3n/2\right]  $; then Ambainis \cite{ambainis:col} and Kutin \cite{kutin}
proved the optimal $\Omega\left(  n^{1/3}\right)  $\ lower bound for functions
$f:\left[  n\right]  \rightarrow\left[  n\right]  $. \ All of these lower
bounds work for both the collision problem and the PTP, though they are
slightly easier to prove for the latter.

The collision problem and the PTP are easily seen to admit Statistical
Zero-Knowledge (SZK) proof protocols. \ Thus, one consequence of the collision
lower bound was the existence of an oracle $A$ such that $\mathsf{SZK}%
^{A}\not \subset \mathsf{BQP}^{A}$.

In talks beginning in 2002,\footnote{See for example: \textit{Quantum Lower
Bounds}, www.scottaaronson.com/talks/lower.ppt; \textit{The Future (and Past)
of Quantum Lower Bounds by Polynomials},
www.scottaaronson.com/talks/future.ppt; \textit{The Polynomial Method in
Quantum and Classical Computing}, www.scottaaronson.com/talks/polymeth.ppt.}
the author often raised the following question:

\begin{quotation}
\noindent\textit{Suppose a function }$f:\left[  n\right]  \rightarrow\left[
n\right]  $\textit{\ is a permutation, rather than far from a permutation.
\ Is there a small (}$\operatorname*{polylog}\left(  n\right)  $%
\textit{-qubit) quantum proof }$\left\vert \varphi_{f}\right\rangle
$\textit{\ of that fact, which can be verified using }$\operatorname*{polylog}%
\left(  n\right)  $\textit{\ quantum queries to }$f$\textit{?}
\end{quotation}

In this paper, we will answer the above question in the negative. \ As a
consequence, we will obtain an oracle $A$ such that $\mathsf{SZK}%
^{A}\not \subset \mathsf{QMA}^{A}$. \ This implies, for example, that any
$\mathsf{QMA}$\ protocol for graph non-isomorphism would need to exploit
something about the problem structure beyond its reducibility to the collision problem.

Given that the relativized $\mathsf{SZK}$\ versus $\mathsf{QMA}$\ problem
remained open for eight years, our solution is surprisingly simple. \ We first
use the in-place amplification procedure of Marriott and Watrous \cite{mw} to
\textquotedblleft eliminate the witness,\textquotedblright\ and reduce the
question to one about quantum algorithms with extremely small acceptance
probabilities. \ We then use a relatively-minor adaptation of the polynomial
degree argument that was used to prove the original collision lower bound.
\ Our proof actually yields an oracle $A$\ such that $\mathsf{SZK}%
^{A}\not \subset \mathsf{A}_{\mathsf{0}}\mathsf{PP}^{A}$, where $\mathsf{A}%
_{\mathsf{0}}\mathsf{PP}$\ is a class defined by Vyalyi \cite{vyalyi} that
sits between $\mathsf{QMA}$\ and $\mathsf{PP}$.

Despite the simplicity of our result, to our knowledge it constitutes
\textit{the first nontrivial lower bound on }$\mathsf{QMA}$\textit{\ query
complexity},\ where \textquotedblleft nontrivial\textquotedblright\ means that
it doesn't follow immediately from earlier results unrelated to $\mathsf{QMA}%
$.\footnote{From the BBBV lower bound for quantum search \cite{bbbv}, one
immediately obtains an oracle $A$ such that $\mathsf{coNP}^{A}\not \subset
\mathsf{QMA}^{A}$: for if there exists a witness state $\left\vert
\varphi\right\rangle $\ that causes a $\mathsf{QMA}$\ verifier to accept the
all-$0$ oracle string, then that same $\left\vert \varphi\right\rangle $\ must
also cause the verifier to accept some string of Hamming weight $1$. \ Also,
since $\mathsf{QMA}\subseteq\mathsf{PP}$ relative to all oracles, the result
of Vereshchagin \cite{vereshchagin}\ that there exists an oracle $A$\ such
that $\mathsf{AM}^{A}\not \subset \mathsf{PP}^{A}$\ implies an $A$ such
that\ $\mathsf{AM}^{A}\not \subset \mathsf{QMA}^{A}$ as well.} \ We hope it
will serve as a starting point for stronger results in the same vein.

\section{Preliminaries\label{PRELIM}}

We assume familiarity with quantum query complexity, as well as with
complexity classes such as $\mathsf{QMA}$ (Quantum Merlin-Arthur),
$\mathsf{QCMA}$ (Quantum Merlin-Arthur with classical witnesses), and
$\mathsf{SZK}$ (Statistical Zero-Knowledge). \ See Buhrman and de Wolf
\cite{bw}\ for a good introduction to quantum query complexity, and the
Complexity Zoo\footnote{www.complexityzoo.com} for definitions of complexity classes.

We now define the main problem we will study. \

\begin{problem}
[Permutation Testing Problem or PTP]Given black-box access to a function
$f:\left[  n\right]  \rightarrow\left[  n\right]  $, and promised that either

\begin{enumerate}
\item[(i)] $f$ is a permutation (i.e., is one-to-one), or

\item[(ii)] $f$ differs from every permutation on at least $n/8$\ coordinates.
\end{enumerate}

The problem is to accept if (i) holds and reject if (ii) holds.
\end{problem}

In the above definition, the choice of $n/8$\ is arbitrary; it could be
replaced by $cn$\ for any $0<c<1$.

As mentioned earlier, Aaronson \cite{aar:col}\ defined the collision problem
as that of deciding whether $f$\ is \textit{one-to-one or two-to-one},
promised that one of these is the case. \ In this paper, we are able to prove
a $\mathsf{QMA}$ lower bound for PTP, but not for the original collision problem.

Fortunately, however, most of the desirable properties of the collision
problem carry over to PTP. \ As an example, we now observe a simple
$\mathsf{SZK}$\ protocol for PTP.

\begin{proposition}
\label{szkprop}PTP has an (honest-verifier) Statistical Zero-Knowledge proof
protocol, requiring $O\left(  \log n\right)  $\ time and\ $O\left(  1\right)
$\ queries to $f$.
\end{proposition}

\begin{proof}
The protocol is the following: to check that $f:\left[  n\right]
\rightarrow\left[  n\right]  $\ is one-to-one, the verifier picks an input
$x\in\left[  n\right]  $\ uniformly at random, sends $f\left(  x\right)  $ to
the prover, and accepts if and only if the prover returns $x$. \ Since the
verifier already knows $x$, it is clear that this protocol has the
zero-knowledge property.

If $f$\ is a permutation, then the prover can always compute $f^{-1}\left(
f\left(  x\right)  \right)  $, so the protocol has perfect completeness.

If $f$ is $n/8$-far from a permutation, then with at least $1/8$\ probability,
the verifier picks an $x$ such that $f\left(  x\right)  $\ has no unique
preimage, in which case the prover can find $x$\ with probability at most
$1/2$. \ So the protocol has constant soundness.
\end{proof}

\subsection{Upper Bounds\label{UPPER}}

To build intuition, we now give a simple $\mathsf{QMA}$\ \textit{upper} bound
for the collision problem. \ Indeed, this will actually be a $\mathsf{QCMA}%
$\ upper bound, meaning that the witness is classical, and only the
verification procedure is quantum.

\begin{theorem}
\label{upperbound}For all $w\in\left[  0,n\right]  $, there exists a
$\mathsf{QCMA}$\ protocol for the collision problem---i.e., for verifying that
$f:\left[  n\right]  \rightarrow\left[  n\right]  $\ is one-to-one rather than
two-to-one---that uses a $w\log n$-bit classical witness and makes $O\left(
\min\left\{  \sqrt{n/w},n^{1/3}\right\}  \right)  $ quantum queries to $f$.
\end{theorem}

\begin{proof}
If $w=O\left(  n^{1/3}\right)  $, then the verifier $V$\ can just ignore the
witness and solve the problem in $O\left(  n^{1/3}\right)  $\ queries\ using
the Brassard-H\o yer-Tapp algorithm \cite{bht}. \ So assume $w\geq Cn^{1/3}$
for some suitable constant $C$.

The witness will consist of claimed values $f^{\prime}\left(  1\right)
,\ldots,f^{\prime}\left(  w\right)  $\ for $f\left(  1\right)  ,\ldots
,f\left(  w\right)  $ respectively. \ Given this witness, $V$ runs the
following procedure.

\begin{enumerate}
\item[(Step 1)] Choose a set of indices $X\subset\left[  w\right]  $\ with
$\left\vert X\right\vert =O\left(  1\right)  $ uniformly at random. \ Query
$f\left(  x\right)  $\ for each $x\in X$, and reject if there is an $x\in X$
such that $f\left(  x\right)  \neq f^{\prime}\left(  x\right)  $.

\item[(Step 2)] Choose a set of indices $Y\subset\left\{  w+1,\ldots
,n\right\}  $ with $\left\vert Y\right\vert =n/w$\ uniformly at random. \ Use
Grover's algorithm to look for a $y\in S$\ such that $f\left(  y\right)
=f^{\prime}\left(  x\right)  $ for some $x\in\left[  w\right]  $. \ If such a
$y$\ is found, then reject; otherwise accept.
\end{enumerate}

Clearly this procedure makes $O\left(  \sqrt{n/w}\right)  $\ quantum queries
to $f$. \ For completeness, notice that if $f$\ is one-to-one, and the witness
satisfies $f^{\prime}\left(  x\right)  =f\left(  x\right)  $\ for all
$x\in\left[  w\right]  $, then $V$ accepts with probability $1$. \ For
soundness, suppose that Step 1 accepts. \ Then with high probability, we have
$f^{\prime}\left(  x\right)  =f\left(  x\right)  $\ for at least (say) a
$2/3$\ fraction of $x\in\left[  w\right]  $. \ However, as in the analysis of
Brassard et al.\ \cite{bht}, this means that, if $f$\ is two-to-one, then with
high probability, a Grover search over $n/w$\ randomly-chosen indices
$y\in\left\{  w+1,\ldots,n\right\}  $\ will succeed at finding a $y$\ such
that $f\left(  y\right)  =f^{\prime}\left(  x\right)  =f\left(  x\right)  $
for some $x\in\left[  w\right]  $. \ So if Step 2 does \textit{not} find such
a $y$, then $V$ has verified to within constant soundness that $f$ is one-to-one.
\end{proof}

For the Permutation Testing Problem, we do not know whether there is a
$\mathsf{QCMA}$\ protocol that satisfies both $T=o\left(  n^{1/3}\right)
$\ and $w=o\left(  n\log n\right)  $. \ However, notice that if $w=\Omega
\left(  n\log n\right)  $, then the witness can just give claimed values
$f^{\prime}\left(  1\right)  ,\ldots,f^{\prime}\left(  n\right)  $ for
$f\left(  1\right)  ,\ldots,f\left(  n\right)  $\ respectively. \ In that
case, the verifier simply needs to check that $f^{\prime}$ is indeed a
permutation, and that $f^{\prime}\left(  x\right)  =f\left(  x\right)  $\ for
$O\left(  1\right)  $ randomly-chosen values $x\in\left[  n\right]  $. \ So if
$w=\Omega\left(  n\log n\right)  $, then the $\mathsf{QMA}$, $\mathsf{QCMA}$,
and $\mathsf{MA}$\ query complexities are all $T=O\left(  1\right)  $.

\section{Main Result\label{MAIN}}

In this section, we prove a lower bound on the $\mathsf{QMA}$\ query
complexity of the Permutation Testing Problem. \ Given a $\mathsf{QMA}%
$\ verifier $V$ for PTP, the first step will be to amplify $V$'s success
probability. \ For this, we use the by-now standard procedure of Marriott and
Watrous \cite{mw}, which amplifies without increasing the size of the quantum witness.

\begin{lemma}
[In-Place Amplification Lemma \cite{mw}]\label{inplace}Let $V$\ be a
$\mathsf{QMA}$\ verifier that uses a $w$-qubit quantum witness, makes $T$
oracle queries, and has completeness and soundness errors $1/3$. Then for all
$s\geq1$, there exists an amplified verifier $V_{s}^{\prime}$\ that uses a
$w$-qubit quantum witness, makes $O\left(  Ts\right)  $\ oracle queries, and
has completeness and soundness errors\ $1/2^{s}$.
\end{lemma}

Lemma \ref{inplace} has a simple consequence that will be the starting point
for our lower bound.

\begin{lemma}
[Guessing Lemma]\label{guesslem}Suppose a language $L$ has a $\mathsf{QMA}%
$\ protocol, which makes $T$\ queries and uses a $w$-qubit quantum witness.
\ Then there is also a quantum algorithm for $L$ (with no witness) that makes
$O\left(  Tw\right)  $ queries, accepts every $x\in L$ with probability at
least $0.9/2^{w}$, and accepts every $x\notin L$ with probability at most
$0.3/2^{w}$.
\end{lemma}

\begin{proof}
Let $V_{s}^{\prime}$\ be the amplified verifier from Lemma \ref{inplace}.
\ Set $s:=w+2$, and consider running $V_{s}^{\prime}$\ with the $w$-qubit
maximally mixed state $I_{w}$ in place of the $\mathsf{QMA}$ witness
$\left\vert \varphi_{x}\right\rangle $. \ Then given any yes-instance $x\in
L$,%
\[
\Pr\left[  V_{s}^{\prime}\left(  x,I_{w}\right)  \text{ accepts}\right]
\geq\frac{1}{2^{w}}\Pr\left[  V_{s}^{\prime}\left(  x,\left\vert \varphi
_{x}\right\rangle \right)  \text{ accepts}\right]  \geq\frac{1-2^{-s}}{2^{w}%
}\geq\frac{0.9}{2^{w}},
\]
while given any no-instance $x\notin L$,%
\[
\Pr\left[  V_{s}^{\prime}\left(  x,I_{w}\right)  \text{ accepts}\right]
\leq\frac{1}{2^{s}}\leq\frac{0.3}{2^{w}}.
\]

\end{proof}

Now let $Q$ be a quantum algorithm for PTP, which makes $T$ queries to $f$.
\ Then just like in the collision lower bound proofs of Aaronson
\cite{aar:col}, Aaronson and Shi \cite{as}, and Kutin \cite{kutin}, the
crucial fact we will need is the so-called \textquotedblleft Symmetrization
Lemma\textquotedblright: namely, $Q$\textit{'s acceptance probability can be
written as a polynomial, of degree at most }$2T$\textit{, in a small number of
integer parameters characterizing }$f$\textit{.}

In more detail, call an ordered pair of integers $\left(  m,a\right)
$\ \textit{valid} if

\begin{enumerate}
\item[(i)] $0\leq m\leq n$,

\item[(ii)] $1\leq a\leq n-m$, and

\item[(iii)] $a$ divides $n-m$.
\end{enumerate}

Then for any valid $\left(  m,a\right)  $, let $S_{m,a}$\ be the set of all
functions $f:\left[  n\right]  \rightarrow\left[  n\right]  $\ that are
one-to-one on $m$\ coordinates and $a$-to-one on the remaining $n-m$%
\ coordinates (with the two ranges not intersecting, so that $\left\vert
\operatorname{Im}f\right\vert =m+\frac{n-m}{a}$). \ The following version of
the Symmetrization Lemma is a special case of the version proved by Kutin
\cite{kutin}.

\begin{lemma}
[Symmetrization Lemma \cite{aar:col,as,kutin}]\label{kutinlem}Let $Q$ be a
quantum algorithm that makes $T$ queries to $f:\left[  n\right]
\rightarrow\left[  n\right]  $. \ Then there exists a real polynomial
$p\left(  m,a\right)  $, of degree at most $2T$, such that%
\[
p\left(  m,a\right)  =\operatorname*{E}_{f\in S_{m,a}}\left[  \Pr\left[
Q^{f}\text{ accepts}\right]  \right]
\]
for all valid $\left(  m,a\right)  $.
\end{lemma}

Finally, we will need a standard result from approximation theory, due to
Paturi \cite{paturi}.

\begin{lemma}
[Paturi \cite{paturi}]\label{polylem}Let $q:\mathbb{R}\rightarrow\mathbb{R}%
$\ be a univariate polynomial such that $0\leq q\left(  j\right)  \leq\delta
$\ for all integers $j\in\left[  a,b\right]  $, and suppose that $\left\vert
q\left(  \left\lceil x\right\rceil \right)  -q\left(  x\right)  \right\vert
=\Omega\left(  \delta\right)  $\ for some $x\in\left[  a,b\right]  $.\ \ Then
$\deg\left(  q\right)  =\Omega\left(  \sqrt{\left(  x-a+1\right)  \left(
b-x+1\right)  }\right)  $.
\end{lemma}

Intuitively, Lemma \ref{polylem} says that $\deg\left(  q\right)
=\Omega\left(  \sqrt{b-a}\right)  $\ if $x$\ is close to one of the endpoints
of the range $\left[  a,b\right]  $, and that $\deg\left(  q\right)
=\Omega\left(  b-a\right)  $\ if $x$ is close to the middle of the range.

We can now prove the $\mathsf{QMA}$\ lower bound for PTP.

\begin{theorem}
[Main Result]\label{mainthm}Let $V$ be a $\mathsf{QMA}$\ verifier for the
Permutation Testing Problem, which makes $T$ quantum queries to the function
$f:\left[  n\right]  \rightarrow\left[  n\right]  $, and which takes a
$w$-qubit quantum witness $\left\vert \varphi_{f}\right\rangle $\ in support
of $f$ being a permutation. \ Then $Tw=\Omega\left(  n^{1/3}\right)  $.
\end{theorem}

\begin{proof}
Assume without loss of generality that $n$ is divisible by $4$. \ Let
$\varepsilon:=0.3/2^{w}$. \ Then by Lemma \ref{guesslem}, from the
hypothesized $\mathsf{QMA}$\ verifier $V$, we can obtain a quantum algorithm
$Q$ for the PTP that makes $O\left(  Tw\right)  $\ queries to $f$, and that
satisfies the following two properties:

\begin{enumerate}
\item[(i)] $\Pr\left[  Q^{f}\text{ accepts}\right]  \geq3\varepsilon$ for all
permutations $f:\left[  n\right]  \rightarrow\left[  n\right]  $.

\item[(ii)] $\Pr\left[  Q^{f}\text{ accepts}\right]  \leq\varepsilon$ for all
$f:\left[  n\right]  \rightarrow\left[  n\right]  $ that are at least
$n/8$-far from any permutation.
\end{enumerate}

Now let $p\left(  m,a\right)  $ be the real polynomial of degree $O\left(
Tw\right)  $ from Lemma \ref{kutinlem}, such that%
\[
p\left(  m,a\right)  =\operatorname*{E}_{f\in S_{m,a}}\left[  \Pr\left[
Q^{f}\text{ accepts}\right]  \right]
\]
for all valid $\left(  m,a\right)  $. \ Then $p$ satisfies the following two properties:

\begin{enumerate}
\item[(i')] $p\left(  m,1\right)  \geq3\varepsilon$ for all $m\in\left[
n\right]  $. \ (For any $f\in S_{m,1}$ is one-to-one on its entire domain.)

\item[(ii')] $0\leq p\left(  m,a\right)  \leq\varepsilon$ for all integers
$0\leq m\leq3n/4$\ and $a\geq2$ such that $a$ divides $n-m$. \ (For in this
case, $\left(  m,a\right)  $\ is valid and every $f\in S_{m,a}$\ is at least
$n/8$-far\ from a permutation.)
\end{enumerate}

So to prove the theorem, it suffices to show that any polynomial $p$
satisfying properties (i') and (ii') above has degree $\Omega\left(
n^{1/3}\right)  $.

Let $g\left(  x\right)  :=p\left(  n/2,2x\right)  $, and let $k$\ be the least
positive integer such that $\left\vert g\left(  k\right)  \right\vert
>2\varepsilon$ (such a $k$ must exist, since $g$\ is a non-constant
polynomial). \ Notice that\ $g\left(  1/2\right)  =p\left(  n/2,1\right)
\geq3\varepsilon$, that $g\left(  1\right)  =p\left(  n/2,2\right)
\leq\varepsilon$, and that $\left\vert g\left(  i\right)  \right\vert
\leq2\varepsilon$\ for all $i\in\left[  k-1\right]  $. \ By Lemma
\ref{polylem}, these facts together imply that $\deg\left(  g\right)
=\Omega\left(  \sqrt{k}\right)  $.

Now let $c:=2k$, and let $h\left(  i\right)  :=p\left(  n-ci,c\right)  $.
\ Then for all integers $i\in\left[  \frac{n}{4c},\frac{n}{c}\right]  $, we
have $0\leq h\left(  i\right)  \leq\varepsilon$, since $\left(  n-ci,c\right)
$\ is valid, $n-ci\leq3n/4$, and $c\geq2$. \ On the other hand, we also have%
\[
h\left(  \frac{n}{2c}\right)  =p\left(  \frac{n}{2},c\right)  =p\left(
\frac{n}{2},2k\right)  =g\left(  k\right)  >2\varepsilon.
\]
By Lemma \ref{polylem}, these facts together imply that $\deg\left(  h\right)
=\Omega\left(  n/c\right)  =\Omega\left(  n/k\right)  $.

Clearly $\deg\left(  g\right)  \leq\deg\left(  p\right)  $ and $\deg\left(
h\right)  \leq\deg\left(  p\right)  $. \ So combining,%
\[
\deg\left(  p\right)  =\Omega\left(  \max\left\{  \sqrt{k},\frac{n}%
{k}\right\}  \right)  =\Omega\left(  n^{1/3}\right)  .
\]

\end{proof}

\section{Oracle Separations\label{ORACLE}}

Using Theorem \ref{mainthm}, we can exhibit an oracle separation between
$\mathsf{SZK}$\ and $\mathsf{QMA}$, thereby answering the author's question
from eight years ago.

\begin{theorem}
\label{szkqmasep}There exists an oracle $A$\ such that $\mathsf{SZK}%
^{A}\not \subset \mathsf{QMA}^{A}$.
\end{theorem}

\begin{proof}
[Proof Sketch]The oracle $A$ will encode an infinite sequence of instances
$f_{n}:\left[  2^{n}\right]  \rightarrow\left[  2^{n}\right]  $\ of the
Permutation Testing Problem, one for each input length $n$. \ Define a unary
language $L_{A}$\ by $0^{n}\in L_{A}$\ if $f_{n}$\ is a permutation, and
$0^{n}\notin L_{A}$\ if $f_{n}$\ is far from a permutation. \ Then Proposition
\ref{szkprop}\ tells us that $L_{A}\in\mathsf{SZK}^{A}$\ for all $A$. \ On the
other hand, Theorem \ref{mainthm}\ tells us that we can choose $A$ in such a
way that $L_{A}\notin\mathsf{QMA}^{A}$, by diagonalizing against all possible
$\mathsf{QMA}$\ verifiers.
\end{proof}

In the rest of the section, we explain how our lower bound actually places
$\mathsf{SZK}$\ outside of a larger complexity class than $\mathsf{QMA}$.
\ First let us define the larger class in question.

\begin{definition}
[Vyalyi \cite{vyalyi}]$\mathsf{A}_{\mathsf{0}}\mathsf{PP}$ is the class of
languages $L$\ for which there exists a $\mathsf{\#P}$\ function $g$, as well
as polynomials $p$\ and $q$, such that for all inputs $x\in\left\{
0,1\right\}  ^{n}$:

\begin{enumerate}
\item[(i)] If $x\in L$\ then $\left\vert g\left(  x\right)  -2^{p\left(
n\right)  }\right\vert \geq2^{q\left(  n\right)  }$.

\item[(ii)] If $x\notin L$\ then $\left\vert g\left(  x\right)  -2^{p\left(
n\right)  }\right\vert \leq2^{q\left(  n\right)  -1}$.
\end{enumerate}
\end{definition}

We now make some elementary observations about $\mathsf{A}_{\mathsf{0}%
}\mathsf{PP}$. \ First, $\mathsf{A}_{\mathsf{0}}\mathsf{PP}$\ is contained in
$\mathsf{PP}$, and contains not only $\mathsf{MA}$\ but also the
slightly-larger class $\mathsf{SBP}$ (Small Bounded-Error Polynomial-Time)
defined by B\"{o}hler et al.\ \cite{bohler}. \ Second, it is not hard to show
that $\mathsf{P}^{\mathsf{P{}romiseA}_{\mathsf{0}}\mathsf{PP}}=\mathsf{P}%
^{\mathsf{PP}}=\mathsf{P}^{\mathsf{\#P}}$. \ The reason is that, by varying
the polynomial $p$, we can obtain a multiplicative estimate of the difference
$\left\vert g\left(  x\right)  -2^{p\left(  n\right)  }\right\vert $, which
then implies that we can use binary search to determine $g\left(  x\right)  $\ itself.

By adapting the result of Aaronson \cite{aar:pp}\ that $\mathsf{PP}%
=\mathsf{PostBQP}$, Kuperberg \cite{kuperberg:jones} gave a beautiful
alternate characterization of $\mathsf{A}_{\mathsf{0}}\mathsf{PP}$ in terms of
quantum computation. \ Let $\mathsf{SBQP}$\ (Small Bounded-Error Quantum
Polynomial-Time) be the class of languages $L$ for which there exists a
polynomial-time quantum algorithm\ that accepts with probability at least
$2^{-p\left(  n\right)  }$\ if $x\in L$, and with probability at most
$2^{-p\left(  n\right)  -1}$\ if $x\notin L$, for some polynomial $p$.

\begin{theorem}
[Kuperberg \cite{kuperberg:jones}]\label{kuperthm}$\mathsf{A}_{\mathsf{0}%
}\mathsf{PP}=\mathsf{SBQP}.$
\end{theorem}

By combining Theorem \ref{kuperthm} with Lemma \ref{inplace}, it is not hard
to reprove the following result of Vyalyi \cite{vyalyi}.

\begin{theorem}
[Vyalyi \cite{vyalyi}]\label{vyalyithm}$\mathsf{QMA}\subseteq\mathsf{A}%
_{\mathsf{0}}\mathsf{PP}$.
\end{theorem}

\begin{proof}
Similar to Lemma \ref{guesslem}. \ Given a language $L$, suppose $L$ has a
$\mathsf{QMA}$\ verifier $V$\ that takes a $w$-qubit quantum witness. \ Then
first apply Marriott-Watrous amplification (Lemma \ref{inplace}), to obtain a
new verifier $V^{\prime}$\ with completeness and soundness errors $0.2/2^{w}$,
which also takes a $w$-qubit quantum witness. \ Next, run $V^{\prime}$ with
the $w$-qubit maximally mixed state $I_{w}$ in place of the witness. \ The
result is a quantum algorithm that accepts every $x\in L$ with probability at
least $0.9/2^{w}$, and accepts every $x\notin L$ with probability at most
$0.2/2^{w}$. \ This implies that $L\in\mathsf{SBQP}$.
\end{proof}

We now observe that our results from Section \ref{MAIN} yield, not only an
oracle $A$ such that $\mathsf{SZK}^{A}\not \subset \mathsf{QMA}^{A}$, but an
oracle $A$ such that $\mathsf{SZK}^{A}\not \subset \mathsf{A}_{\mathsf{0}%
}\mathsf{PP}^{A}$, which is a stronger separation.

\begin{theorem}
\label{szka0pp}There exists an oracle $A$\ such that $\mathsf{SZK}%
^{A}\not \subset \mathsf{A}_{\mathsf{0}}\mathsf{PP}^{A}$.
\end{theorem}

\begin{proof}
[Proof Sketch]As in Theorem \ref{szkqmasep}, the oracle $A$ encodes an
infinite sequence of instances $f_{n}:\left[  2^{n}\right]  \rightarrow\left[
2^{n}\right]  $\ of the Permutation Testing Problem. \ The key observation is
that Theorem \ref{mainthm} rules out, not merely any $\mathsf{QMA}$\ protocol
for PTP, but also any $\mathsf{SBQP}$\ algorithm: that is, any polynomial-time
quantum algorithm that accepts with probability at least $2\varepsilon$\ if
$f_{n}$\ is a permutation, and with probability at most $\varepsilon$\ if
$f_{n}$\ is far from a permutation, for some $\varepsilon>0$. \ This means
that we can use Theorem \ref{mainthm}\ to diagonalize against $\mathsf{SBQP}%
$\ (or equivalently $\mathsf{A}_{\mathsf{0}}\mathsf{PP}$) machines.
\end{proof}

\section{Open Problems\label{OPEN}}

\begin{enumerate}
\item[(1)] It is strange that our lower bound works only for the Permutation
Testing Problem, and not for the original collision problem (i.e., for
certifying that $f$ is one-to-one rather than two-to-one). \ Can we rule out
succinct quantum proofs for the latter?

\item[(2)] Even for PTP, there remains a large gap between the upper and lower
bounds that we can prove on $\mathsf{QMA}$\ query complexity. \ Recall that
our lower bound has the form $Tw=\Omega\left(  n^{1/3}\right)  $, where
$T$\ is the query complexity and $w$ is the number of qubits in the witness.
\ By contrast, if $w=o\left(  n\log n\right)  $, then we do not know of
\textit{any} $\mathsf{QMA}$\ protocol that achieves $T=o\left(  n^{1/3}%
\right)  $---i.e., that does better than simply ignoring the witness and
running the Brassard-H\o yer-Tapp algorithm. \ It would be extremely
interesting to get sharper results on the tradeoff between $T$ and $w$. \ (As
far as we know, it is open even to get a sharp tradeoff for \textit{classical}
$\mathsf{MA}$\ protocols.)

\item[(3)] For the collision problem, the PTP, or any other black-box problem,
is there a gap (even just a polynomial gap) between the $\mathsf{QMA}$\ query
complexity and the $\mathsf{QCMA}$\ query complexity? \ This seems like a
difficult question, since currently, the one lower bound technique that we
have for\ $\mathsf{QCMA}$---namely, the reduction to $\mathsf{SBQP}%
$\ exploited in this paper---\textit{also} works for $\mathsf{QMA}$. \ It
follows that a new technique will be needed to solve the old open problem of
constructing an oracle $A$\ such that $\mathsf{QCMA}^{A}\neq\mathsf{QMA}^{A}$.
\ (Currently, the closest we have is a \textit{quantum} oracle separation
between $\mathsf{QMA}$\ and $\mathsf{QCMA}$, shown by Aaronson and Kuperberg
\cite{ak}.)

\item[(4)] Watrous (personal communication) asked whether there exists an
oracle $A$ such that $\mathsf{SZK}^{A}\not \subset \mathsf{PP}^{A}$. \ Since
$\mathsf{PP}\subseteq\mathsf{P}^{\mathsf{P{}romiseA}_{\mathsf{0}}\mathsf{PP}}%
$, our oracle separation between\ $\mathsf{SZK}$ and $\mathsf{A}_{\mathsf{0}%
}\mathsf{PP}$ comes \textquotedblleft close\textquotedblright\ to answering
Watrous's question. \ However, a new technique seems needed to get from
$\mathsf{A}_{\mathsf{0}}\mathsf{PP}$\ to $\mathsf{PP}$.
\end{enumerate}

\bibliographystyle{plain}
\bibliography{thesis}

\begin{thebibliography}{10}

\bibitem{aar:col}
S.~Aaronson.
\newblock Quantum lower bound for the collision problem.
\newblock In {\em Proc. ACM STOC}, pages 635--642, 2002.
\newblock quant-ph/0111102.

\bibitem{aar:pp}
S.~Aaronson.
\newblock Quantum computing, postselection, and probabilistic polynomial-time.
\newblock {\em Proc. Roy. Soc. London}, A461(2063):3473--3482, 2005.
\newblock quant-ph/0412187.

\bibitem{ak}
S.~Aaronson and G.~Kuperberg.
\newblock Quantum versus classical proofs and advice.
\newblock {\em Theory of Computing}, 3(7):129--157, 2007.
\newblock Previous version in Proceedings of CCC 2007. quant-ph/0604056.

\bibitem{as}
S.~Aaronson and Y.~Shi.
\newblock Quantum lower bounds for the collision and the element distinctness
  problems.
\newblock {\em J. ACM}, 51(4):595--605, 2004.

\bibitem{ambainis:col}
A.~Ambainis.
\newblock Polynomial degree and lower bounds in quantum complexity: collision
  and element distinctness with small range.
\newblock {\em Theory of Computing}, 1:37--46, 2005.
\newblock quant-ph/0305179.

\bibitem{bbbv}
C.~Bennett, E.~Bernstein, G.~Brassard, and U.~Vazirani.
\newblock Strengths and weaknesses of quantum computing.
\newblock {\em SIAM J. Comput.}, 26(5):1510--1523, 1997.
\newblock quant-ph/9701001.

\bibitem{bohler}
E.~B\"{o}hler, C.~Gla\ss er, and D.~Meister.
\newblock Error-bounded probabilistic computations between {MA} and {AM}.
\newblock {\em J. Comput. Sys. Sci.}, 72(6):1043--1076, 2006.

\bibitem{bht}
G.~Brassard, P.~H{\o}yer, and A.~Tapp.
\newblock Quantum algorithm for the collision problem.
\newblock {\em ACM SIGACT News}, 28:14--19, 1997.
\newblock quant-ph/9705002.

\bibitem{bw}
H.~Buhrman and R.~de Wolf.
\newblock Complexity measures and decision tree complexity: a survey.
\newblock {\em Theoretical Comput. Sci.}, 288:21--43, 2002.

\bibitem{kuperberg:jones}
G.~Kuperberg.
\newblock How hard is it to approximate the {J}ones polynomial?
\newblock 2009.
\newblock arXiv:0908.0512.

\bibitem{kutin}
S.~Kutin.
\newblock Quantum lower bound for the collision problem with small range.
\newblock {\em Theory of Computing}, 1:29--36, 2005.
\newblock quant-ph/0304162.

\bibitem{mw}
C.~Marriott and J.~Watrous.
\newblock Quantum {A}rthur-{M}erlin games.
\newblock {\em Computational Complexity}, 14(2):122--152, 2005.

\bibitem{paturi}
R.~Paturi.
\newblock On the degree of polynomials that approximate symmetric {B}oolean
  functions.
\newblock In {\em Proc. ACM STOC}, pages 468--474, 1992.

\bibitem{vereshchagin}
N.~Vereshchagin.
\newblock On the power of {PP}.
\newblock In {\em Proc. IEEE Conference on Computational Complexity}, pages
  138--143, 1992.

\bibitem{vyalyi}
M.~Vyalyi.
\newblock {QMA=PP} implies that {PP} contains {PH}.
\newblock ECCC TR03-021, 2003.

\end{thebibliography}

\end{document}